\documentclass[12pt,reqno]{article}

\usepackage[usenames]{color}
\usepackage{amssymb}
\usepackage{amsmath}
\usepackage{amsthm}
\usepackage{amsfonts}
\usepackage{amscd}
\usepackage{graphicx}
\usepackage{lscape}

\usepackage[colorlinks=true,
linkcolor=webgreen,
filecolor=webbrown,
citecolor=webgreen]{hyperref}

\definecolor{webgreen}{rgb}{0,.5,0}
\definecolor{webbrown}{rgb}{.6,0,0}

\usepackage{color}
\usepackage{fullpage}
\usepackage{float}

\usepackage{graphics}
\usepackage{latexsym}
\usepackage{epsf}
\usepackage{breakurl}

\newcommand{\seqnum}[1]{\href{https://oeis.org/#1}{\rm \underline{#1}}}
\def\suchthat{\, : \, }
\DeclareMathOperator{\cmp}{cmp}
\DeclareMathOperator{\link}{link}
\DeclareMathOperator{\pref}{pref}

\DeclareMathOperator{\boolean}{boolean}
\DeclareMathOperator{\faceq}{faceq}
\DeclareMathOperator{\rtspec}{rtspec}
\DeclareMathOperator{\power}{power3}
\DeclareMathOperator{\bnd}{bnd}
\DeclareMathOperator{\per}{per}
\DeclareMathOperator{\differ}{differ}

\begin{document}

\theoremstyle{plain}
\newtheorem{theorem}{Theorem}
\newtheorem{corollary}[theorem]{Corollary}
\newtheorem{lemma}[theorem]{Lemma}
\newtheorem{proposition}[theorem]{Proposition}

\theoremstyle{definition}
\newtheorem{definition}[theorem]{Definition}
\newtheorem{example}[theorem]{Example}
\newtheorem{conjecture}[theorem]{Conjecture}

\theoremstyle{remark}
\newtheorem{remark}[theorem]{Remark}

\title{Properties of a Class of Toeplitz Words}

\author{Gabriele Fici\\
Dipartimento di Matematica e Informatica\\
Universit\`a di Palermo\\
Italy\\
\href{mailto:gabriele.fici@unipa.it}{\tt gabriele.fici@unipa.it}\\
\ \\
Jeffrey Shallit\footnote{Research funded by a grant from NSERC, 2018-04118.}\\
School of Computer Science\\
University of Waterloo\\
Waterloo, ON  N2L 3G1\\
Canada\\
\href{mailto:shallit@uwaterloo.ca}{\tt shallit@uwaterloo.ca}
}

\maketitle

\begin{abstract}
We study the properties of the uncountable set of
{\it Stewart words}.   These are 
Toeplitz words specified by 
infinite sequences of Toeplitz patterns of the form $\alpha\beta\gamma$, where
$\alpha,\beta,\gamma$ is any permutation of the symbols {\tt 0,1,?}.
We determine the critical exponent of the Stewart words, prove
that they avoid the pattern $xxyyxx$, find all factors that
are palindromes, and determine their
subword complexity.  An interesting aspect of our work is
that we use automata-theoretic methods and a decision procedure
for automata to carry out the proofs.

\bigskip

\end{abstract}

\section{Introduction}

{\it Toeplitz words\/} are infinite words formed by
an iterative process, specified by a sequence  $(t_i)_{i\geq 1}$
of one or more 
{\it Toeplitz patterns} \cite{Jacobs&Keane:1969,Cassaigne&Karhumaki:1997}.   A Toeplitz pattern is a finite
word over the alphabet $\Sigma \, \cup \, \{ \mbox{\tt ?} \}$,
where $\Sigma$ is a finite alphabet and
{\tt ?} is a distinguished symbol.

We start by
constructing the infinite word 
$t_1^\omega = t_1t_1t_1\cdots$ by
repeating $t_1$ infinitely.   Next, we consider the locations of
the {\tt ?} symbols in $t_1^\omega$ and replace all of the
corresponding {\tt ?} symbols by $t_2^\omega$, and so forth.

As an example, consider using only the single Toeplitz pattern {\tt 0?1?}.
At the first stage we get
$$ \mbox{{\tt 0?1?0?1?0?1?0?1?0?1?0?1?0?1?0?1?}} \cdots .$$
At the second stage we get
$$ \mbox{{\tt 001?011?001?011?001?011?001?011?} }\cdots .$$
At the third stage we get
$$ \mbox{{\tt 0010011?0011011?0010011?0011011?} }\cdots .$$
At the fourth stage we get
$$ \mbox{{\tt 001001100011011?001001110011011?} }\cdots.$$
The limit of this process yields the infinite word
$$ \mbox{{\tt 00100110001101100010011100110110} }\cdots,$$
known as the {\it regular paperfolding
word}.

On the other hand, instead of a single pattern, at each stage we
could choose between the two patterns {\tt 0?1?} and {\tt 1?0?}.
In this case we get an uncountable set of infinite words,
called the {\it paperfolding words}.  They share many of the
same properties; see, for example,
\cite{Dekking&MendesFrance&vanderPoorten:1982}.

In this paper we consider another uncountable set of Toeplitz words that we call {\it Stewart words}.   Here we can choose, at each stage, any of the six pattern of the
form $\alpha\beta\gamma$, where $\{\alpha,\beta,\gamma\}$ is some permutation of the
three symbols {\tt 0,1,?} and that we call 
{\it Stewart patterns\/}: 
\begin{align*}
{\tt a} &= \mbox{\tt 01?}; &\quad
{\tt b} &= \mbox{\tt 10?}; \\
{\tt c} &= \mbox{\tt 0?1}; &\quad
{\tt d} &= \mbox{\tt 1?0}; \\
{\tt e} &= \mbox{\tt ?01}; &\quad
{\tt f} &= \mbox{\tt ?10}.
\end{align*}

Some choices of these patterns correspond to well-known
sequences.  We give two examples now.
\begin{itemize}
\item
The pattern sequence ${\tt c}^\omega = {\tt ccc} \cdots$
specifies the so-called {\it Stewart choral sequence} \cite{Stewart:2006}
$$ {\tt 001001011} \cdots,$$
the fixed point of the morphism
${\tt 0} \rightarrow {\tt 001}, \quad {\tt 1} \rightarrow {\tt 011}$.
This is sequence \seqnum{A116178} in the
{\it On-Line Encyclopedia of Integer Sequences}
(OEIS) \cite{Sloane:2021}.   This explains our choice of the term ``Stewart word''.
It was also mentioned by Cassaigne and
Karhum{\"aki} \cite[Example 4]{Cassaigne&Karhumaki:1997} and
Berstel and Karhum{\"a}ki
\cite[pp.~196--197]{Berstel&Karhumaki:2003}.

\item The pattern sequence ${\tt (ab)}^\omega = {\tt ababab} \cdots$ specifies the so-called
{\it Sierpi\'nski gasket word}
$$ {\tt 011010010} \cdots,$$
the fixed point of the morphism
${\tt 0} \rightarrow {\tt 011}, \quad {\tt 1} \rightarrow {\tt 010}$.
This is sequence \seqnum{A156595} in the
{\it On-Line Encyclopedia of Integer Sequences}
(OEIS) \cite{Sloane:2021}.

\item The eight ``generalized choral sequences'' of Noche \cite{Noche:2008,Noche:2008b,Noche:2011}
are (up to renaming the first letter) 
the six sequences specified by
${\tt a}^\omega$, 
${\tt b}^\omega$,
${\tt c}^\omega$, 
${\tt d}^\omega$,
${\tt e}^\omega$, 
${\tt f}^\omega$.\footnote{In the case that a Stewart word is specified by an infinite word with a suffix in $\{ {\tt e,f} \}^\omega$, and only in this case, the limit of the Toeplitz construction is an infinite word containing a single occurrence of {\tt ?}. In this special case, there are two distinct Stewart words, obtained by replacing this single occurrence of {\tt ?} with $\tt 0$ and $\tt 1$, respectively.} For these specific sequences, Noche proved some of the same properties we discuss here.
\end{itemize}

The Stewart words, specified by {\it arbitrary\/} infinite sequences of Stewart
patterns in 
$$\{ {\tt a,b,c,d,e,f} \}^\omega,$$
share many properties, as conjectured
by the first author \cite{Fici:2021}.    
In this article we prove these properties, and others.

An interesting feature of our proofs is that they are based
on automata and expressing the desired property in first-order logic, and
then using a decision procedure to prove the
result ``automatically''.   We emphasize that this method permits us to
handle sequences generated by 
{\it all\/} infinite sequences of
the Stewart patterns, not just the periodic ones.   The same idea
was used previously to analyze the paperfolding sequences
\cite{Goc&Mousavi&Schaeffer&Shallit:2015}.

\section{Notation}

We define $\Sigma_k = \{ {\tt 0,1,} \ldots, k-1 \}$.  We use the concept of {\it deterministic finite automaton with output}, the
DFAO, from \cite{Allouche&Shallit:2003}.   This is a 
6-tuple $(Q, \Sigma, \Delta, \delta, q_0, \tau)$, where
the output on input $x$ is $\tau(\delta(q_0,x))$.   For
an infinite word ${\bf t} = t_0 t_1 t_2 \cdots$ we
define ${\bf t}[i] = t_i$, and the analogous notation for finite
words.  We write ${\bf t}[i..j] = t_i t_{i+1} \cdots t_j$.

For two infinite words ${\bf t} = t_0 t_1 t_2 \cdots$ and
${\bf u} = u_0 u_1 u_2 \cdots$ we define
${\bf t} \times {\bf u}$ to be the infinite word
$[t_0, u_0] [t_1, u_1] [t_2, u_2] \cdots$, and the same
for two finite words of the same length.   We use italic letters for finite
words and bold-face for infinite words.

\section{Finite Toeplitz words and the Stewart automaton}

We start by considering {\it finite\/} sequences of Stewart patterns
and the {\it finite\/} prefixes of the Stewart words they specify.

Let $t$ be a finite word over the alphabet
$A := \{ {\tt a,b,c,d,e,f } \}$.  Such a finite word defines
an infinite Stewart word $s(t)$
over the alphabet $B := \{ \mbox{\tt 0,1,?} \}$
as explained above.   However, we will find it useful to instead define
the map
$ T(t)$ to be the
{\it finite prefix} of length $3^{|t|}$ of $s(t)$.
Thus, for example, 
$T({\tt af}) = {\tt 01?011010}$ and
$T({\tt afe}) = {\tt 01?011010010011010011011010}$.

Alternatively, for a finite word $t \in A^*$ we can define $T(t)$ as follows:
We have $T(\epsilon) = \mbox{\tt ?}$;  given $y = T(t)$, we
define $T(tg)$ for a single letter $g \in A$ by replacing the three
\mbox{\tt ?} symbols in $y^3$, in order,  by the three symbols $T(g)$.  

\begin{theorem}
There is a 3-state DFAO $M$ over the alphabet
$A \times \Sigma_3$ that, takes as input a word where
the first component spells out the finite sequence $t$
of Stewart patterns,
and the second component spells out $n$, expressed in base 3 (with
the least significant digit first), and computes $T(t)[n]$,
the $n$'th symbol of $T(t)$, where we index
starting at position $0$.  Here the name of the state is also
its output.
\end{theorem}

\begin{proof}
The DFAO $M$ is depicted in Figure~\ref{fig1}.  Asterisks match all possible inputs.
\begin{figure}[H]
\begin{center}
\includegraphics[width=4in]{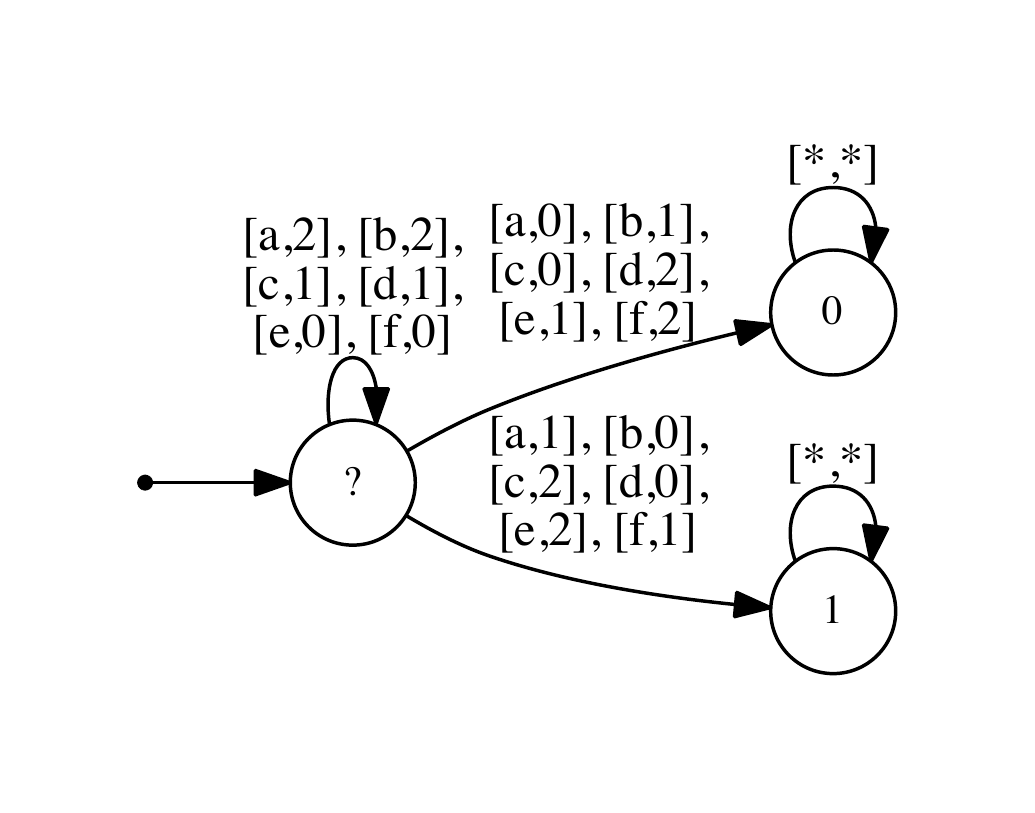}
\end{center}
\caption{The Stewart automaton $M$.}
\label{fig1}
\end{figure}
To see that the Stewart automaton works as claimed, an easy induction shows
that if $n = \sum_{0 \leq i < r} e_i 3^i$ with $e_i \in \Sigma_3$, and
$\delta$ is the transition function of $M$,
then 
$$\delta(\mbox{\tt ?}, [t_0, e_0][t_1, e_1] \cdots
[t_{r-1}, e_{r-1}]) = T(t_0 t_1 \cdots t_{r-1}) [n] .$$
\end{proof}

Let us agree to call $\tt 0$ and $\tt 1$ {\it boolean\/} symbols, in
contrast with the {\tt ?} symbol.   A factor of a finite Toeplitz word is called boolean if all of its symbols are boolean.

\section{Decision procedures and {\tt Walnut}}

The strategy in this paper is to express our conjectures as first-order logic formulas about finite Toeplitz words $T(t)$.  Because we have an automaton computing $T(t)$, these formulas can be proved or disproved using a well-known extension of Presburger arithmetic \cite{Bruyere&Hansel&Michaux&Villemaire:1994}.

We use {\tt Walnut} \cite{Mousavi:2016}, a free software theorem-prover for first-order statements about automata and the sequences they compute.   For more about {\tt Walnut},
see \url{https://cs.uwaterloo.ca/~shallit/walnut.html}.

Here is a brief guide to the syntax of {\tt Walnut}.  All variables
refer to natural numbers.   
\begin{itemize}
    \item {\tt ?lsd\_}$x$ instructs {\tt Walnut} to evaluate the formula where the default base of representation of integers is $x$.  In our paper $x$ can be either $3$ or $7$.  This instruction has local scope.
    
    \item {\tt A} represents the universal quantifier $\forall$;
    \item {\tt E} represents the existential quantifier $\exists$;
    \item {\tt \&} represents logical ``and'' ($\wedge$);
    \item {\tt |} represents logical ``or'' ($\vee$);
    \item {\tt eval} evaluates a logical formula with no free variables and decides if it is true or false;
    \item {\tt def} defines an automaton for a logical formula
    for later use;
    \item {\tt reg} allows one to define an automaton from a 
    regular expression, for later use;
    \item {\tt TP[t][n]} represents $T(t)[n]$.
\end{itemize}

\section{Implementing the Stewart automaton in {\tt Walnut}}

When we implement the DFAO $M$ in {\tt Walnut}, some minor modifications are needed.  Because the outputs of automata in {\tt Walnut} must be integers and not letters, 
we have to change the coding of letters from {\tt a,b,c,d,e,f} to
numbers; we use $1,2,3,4,5,6$ instead.   We do not use $0$ as one of the six letter codes, in order that
finite Stewart pattern sequences be allowed to end in arbitrarily many $0$'s; this 
permits an automaton
to compare pattern sequences of different lengths (by padding the shorter with 
trailing zeros) and to determine if one is a prefix of another.   Therefore, a finite Stewart pattern sequence $t$ can be written as any element of $t0^*$.
The automaton $M$ in {\tt Walnut}
is stored under the name {\tt TP.txt} in the
{\tt Word Automata Library}.  We identify a finite Stewart pattern sequence with the integer it represents in base $7$, remembering that all integers in this paper are represented starting with the least significant digit.

We also have to change the output alphabet of the automaton, replacing {\tt ?} with $2$.  This gives the automaton in Figure~\ref{fig6}.  A state is labeled in the form ``state name/output''.  Here
$F = \{1,2,3,4,5,6\}$ and $G = \Sigma_3$.
\begin{figure}[htb]
\begin{center}
\includegraphics[width=4in]{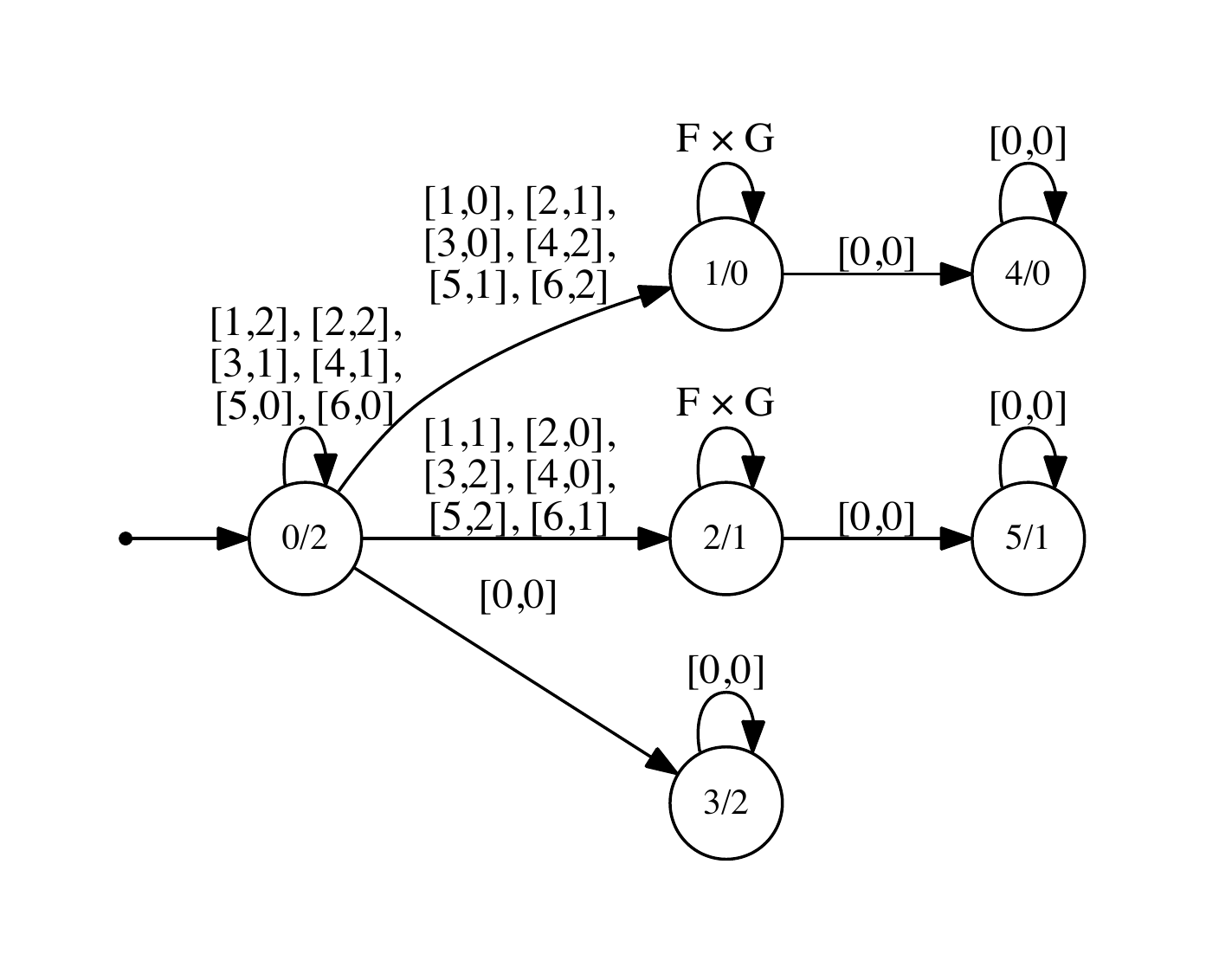}
\end{center}
\caption{The automaton {\tt TP.txt} in {\tt Walnut}.}
\label{fig6}
\end{figure}

One of our main ideas is that to prove many statements about sequences generated by infinite sequences
of Stewart patterns, it suffices to study {\it finite\/}
sequences of Stewart patterns only.  This is implied by the following lemma.
\begin{lemma}
Let $t_1 t_2 \cdots$ be an infinite sequence of Stewart patterns.
Then, for all $r \geq 1$, there exists a letter
$h \in A$ such that the prefix of
$T(t_1 t_2 \cdots t_{r-1} h)$ of 
length $3^{r-1}$ is a prefix of $T(t_1 t_2 \cdots)$.
\end{lemma}

\begin{proof}
There is only one {\tt ?} symbol in $T(t_1 t_2\cdots t_{r-1})$ that
can be substituted by later patterns, and whatever substitution
is made is accounted for by a letter of $A$.
\end{proof}

The lemma tells us that if we study the finite sequences arising from
Stewart patterns of length $r$, then the first $3^{r-1}$ symbols
of such sequences cover all possibilities for the prefixes of
length $3^{r-1}$ of the sequences arising from all infinite sequences
of Stewart patterns.

Even more than this, we can prove the following:
\begin{theorem}
Let $n\geq 1, \ell \geq 2$ be integers with $n \leq 3^{\ell -2}$, and 
let $t\in A^\ell$ be a length-$\ell$ sequence of Stewart codes.  If
$t$ is a prefix of $u$, 
then every length-$n$ boolean factor
that appears in $T(u)$ also appears in $T(t)$.
\label{thm3}
\end{theorem}

\begin{proof}
We can prove this with {\tt Walnut}.  We need to define various formulas
useful in expressing the assertion of Theorem~\ref{thm3} in first-order logic:
\begin{itemize}

    \item $\faceq(i,j,n,t)$ expresses the assertion that
    $T(t)[i..i+n-1]=T(t)[j..j+n-1]$, for integers
    $i, j, n$ and a finite sequence of Stewart patterns
    $t$.
    
    \item $\cmp(i, j, n, t, u)$ expresses the assertion that
    $T(t)[i..i+n-1] = T(u)[j..j+n-1]$, for integers
    $i, j, n$ and finite sequences of Stewart patterns
    $t$ and $u$.
    
    \item $\boolean(i,n,t)$ expresses the assertion that
    $T(t)[i..i+n-1]$ is boolean (contains no {\tt ?}).
    
    \item $\pref(t, u)$ expresses the assertion that
    $t$ is a prefix of $u$.  For the purpose of this
    relationship, trailing zeros in $t$ and $u$ are ignored.
    
    \item $\link(x,t)$ expresses the assertion that $x$ is an
    integer satisfying $x = 3^{|t|}$.   This links the base-$3$
    representation of $x$ with the base-$7$ representation of $t$.
    
    \item $\bnd(x,y)$ expresses the assertion that
    $y = 3^{\lceil \log_3 x \rceil}$ for $x \geq 1$.
    
    \item $\power(x)$ expresses the assertion that $x \geq 1$
    is a power of $3$.
\end{itemize}

Here is the implementation of these subroutines in {\tt Walnut}.
\begin{verbatim}
def faceq "?lsd_3 Ak (k<n) => TP[t][i+k]=TP[t][j+k]":
# vars i,j,n,t
# 88 states, 28965 ms

def cmp "?lsd_3 Ak (k<n) => TP[t][i+k]=TP[u][j+k]":
# vars i,j,n,t,u
# 431 states, 76 Gigabytes of RAM and 2014067ms CPU

def boolean "?lsd_3 Ak (k<n) => (TP[t][i+k]=@0|TP[t][i+k]=@1)":
# vars i,n,t
# 24 states

reg pref lsd_7 lsd_7 "([1,1]|[2,2]|[3,3]|[4,4]|[5,5]|[6,6])*
([0,1]|[0,2]|[0,3]|[0,4]|[0,5]|[0,6])*[0,0]*":
# vars t1,t2
# 3 states

reg link lsd_3 lsd_7 "([0,1]|[0,2]|[0,3]|[0,4]|[0,5]|[0,6])*[1,0][0,0]*":
# vars x,t
# 2 states

reg bnd lsd_3 lsd_3 "([0,0]|[1,0]|[2,0])*[0,1][0,0]*|[0,0]*[1,1][0,0]*":
# vars x,y
# 3 states

reg power3 lsd_3 "0*10*":
# var x
\end{verbatim}

Finally, we translate the statement of Theorem~\ref{thm3} into
{\tt Walnut}.  Instead of representing $|t|$ (resp., $|u|$), we use the variable $x$ (resp., $y$) to represent $3^{|t|}$ (resp., $3^{|u|}$).
\begin{verbatim}
eval thm3 "?lsd_3 An,t,u,x,y,j (n>=1 & $link(x,t) & $link(y,u) & 
$pref(t,u) & n<=x/9 & j+n<=y & $boolean(j,n,u)) => Ei $cmp(i,j,n,t,u)":
# returns TRUE
# 24534 ms
\end{verbatim}
This returns {\tt TRUE}.  The expensive part of the computation is
the construction of the {\tt cmp} automaton, which required 76 Gigabytes of RAM and
1886 seconds of CPU time on a Linux machine. It has 431 states.
\end{proof}

\begin{remark}
It turns out that the bound $3^{\ell -2}$ is optimal for $n \geq 2$;
this can also be proved with {\tt Walnut}.  We omit the details.
\end{remark}

The rest of the paper uses the following strategy.  Instead of directly proving
results about an infinite Stewart word $T({\bf t})$ (which would require
B\"uchi automata), we must often instead {\it restate\/} our claims about infinite
Stewart words to be about their finite prefixes generated by a finite
prefix $t$ of $\bf t$.  This sometimes requires creative rephrasing involving
stronger statements.  For
example, instead of proving a particular kind of factor appears in
$T({\bf t})$,
we prove that it appears provided we look at a sufficiently long prefix, where
the length of the prefix is given explicitly in terms of the size of the factor.
This is where the estimate in Theorem~\ref{thm3} becomes so useful.  An
example is Lemma~\ref{lemma7} below.

\section{Palindromic factors of the Stewart words}

We are interested in what factors of Stewart words are
(nonempty) palindromes.
To do so, we write a first-order logical statement for the lengths
of such words.
$$ \exists t,i \ \forall j\ (j<n) \implies T(t)[i+j]=T(t)[i+n-j-1] .$$

When we implement this in {\tt Walnut} we get
\begin{verbatim}
def pal "?lsd_3 Et,i Aj (j<n) => TP[t][i+j]=TP[t][(i+n)-(j+1)]":
# 4 states, 471 ms
\end{verbatim}

The resulting automaton is depicted below in Figure~\ref{fig2}.
It gives the base-3 representation of the possible lengths of palindromes
occurring in the Stewart words, starting with the least significant digit.
\begin{figure}[H]
\begin{center}
\includegraphics[width=4in]{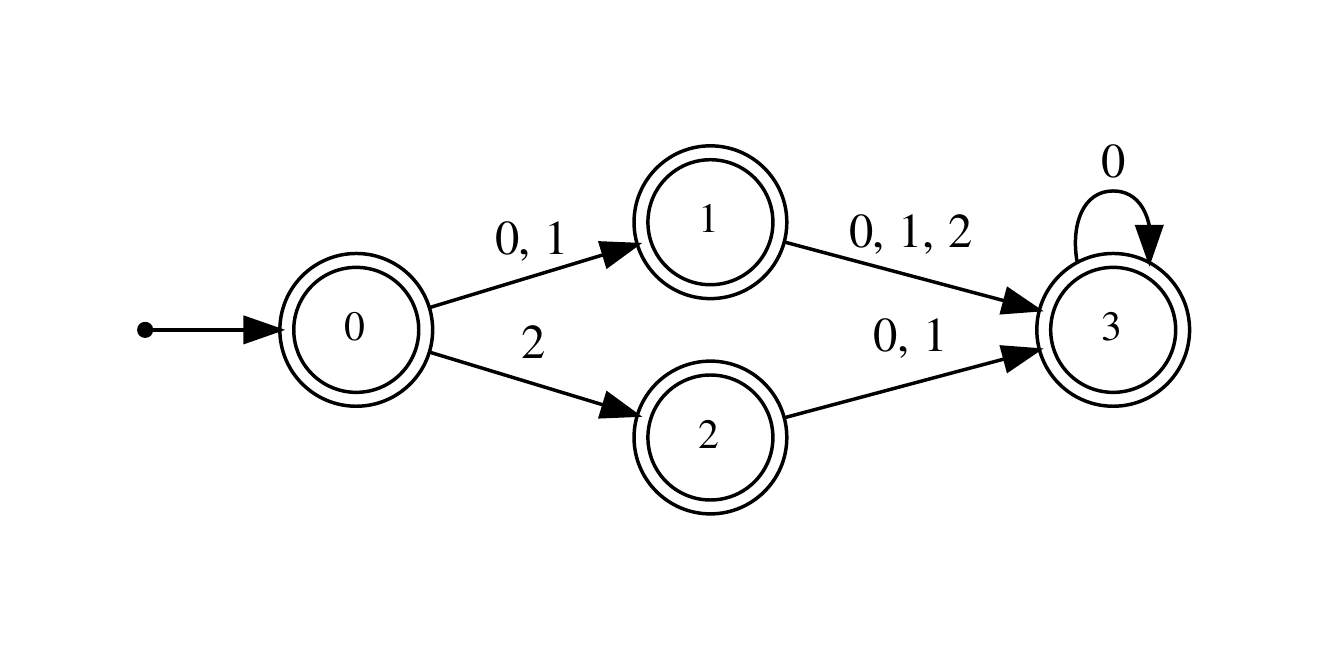}
\end{center}
\caption{Lengths of palindromes in the Stewart words.}
\label{fig2}
\end{figure}
As you can see from inspecting this automaton, the possible lengths
of palindromes are exactly $\{ 0,1,2,3,4,5,6,7 \}$.
Now we can prove
\begin{theorem}
The only possible palindromes occurring in the Stewart words
are as follows:
$$ \{ \epsilon, {\tt 0}, {\tt 1}, {\tt ?}, {\tt 00}, {\tt 11}, {\tt 010}, {\tt 101}, {\tt 0110}, {\tt 1001}, {\tt 00100}, {\tt 11011}, {\tt 010010},
{\tt 101101}, {\tt 0110110}, {\tt 1001001} \}.$$
Furthermore, each such palindrome occurs in the length-$81$ word specified by any
Stewart pattern sequence of length 4.
\end{theorem}

\section{The pattern $xxyyxx$}

Let us consider occurrences of the pattern $xxyyxx$ in Stewart words.   Here we assume that $x$ is nonempty, but
$y$ is allowed to be empty.

Say that $|x|=m$ and $|y|=n$.    
We can express the assertion that such a pattern occurs as follows:
\begin{align*}
& \exists t,i,m,n \ \ (m \geq 1) \ \wedge \\
& (\forall j \ (j<m) \implies T(t)[i+j]=T(t)[i+j+m]) \ \wedge \\
& (\forall j \ (j<2m) \implies T(t)[i+j]=T(t)[i+j+2m+2n] \ \wedge \\
& (\forall j \ (j<n) \implies T(t)[i+j+2m]=T(t)[i+j+2m+n] .
\end{align*}

The translation into {\tt Walnut} is
\begin{verbatim}
eval xxyyxx "?lsd_3 Et,i,m,n (m>=1) &
(Aj (j<m) => TP[t][i+j]=TP[t][i+j+m]) &
(Aj (j<2*m) => TP[t][i+j]=TP[t][i+j+2*m+2*n]) &
(Aj (j<n) => TP[t][i+j+2*m]=TP[t][i+j+2*m+n])":
# returns FALSE
# 37396 ms
\end{verbatim}
and {\tt Walnut} returns {\tt FALSE}, so there is no such pattern.

\section{Critical exponent of the Stewart words}

Let us recall the notions of period and exponent of a finite word.
If $w[i]=w[i+p]$ for $p \geq 1$ and $0 \leq i < |w|-p$, then
$w$ is said to have period $p$.   The smallest period
of $w$ is denoted $\per(w)$.   The exponent of $w$,
written $\exp(w)$, is defined to be $|w|/\per(w)$.
\begin{theorem}
Let ${\bf t}$ be an infinite sequence of Stewart
patterns.  Then the Stewart word $T({\bf t})$ has critical exponent
$3$, which is not attained by any finite factor.
\label{thm5}
\end{theorem}

\begin{proof}
Let us prove first that none of the infinite words contains
a cube.

We write down a first-order logical formula for the property
that some finite Stewart word coded by $t$ contains a cube, namely
$$ \exists i, p, t \ (p\geq 1) \ \wedge \ \forall j \ (j<2p) \implies
	T(t)[i+j]=T(t)[i+p+j] .$$

The translation into {\tt Walnut} is as follows:
\begin{verbatim}
eval hascube "?lsd_3 Ei,p,t p>=1 & Aj (j<2*p) => TP[t][i+j]=TP[t][i+j+p]":
# returns FALSE
# 1280 ms
\end{verbatim}
When we run this in {\tt Walnut}, we get the answer
{\tt FALSE}, so there are no cubes.  In fact,
this was already known, and
is a consequence of a recent result of Boccuto and Carpi \cite{Boccuto&Carpi:2020}.

Next, given an infinite sequence of Stewart patterns $\bf t$,
we have to demonstrate the existence of factors of exponent arbitrarily
close to $3$. 
It suffices to prove the following result.

\begin{lemma}
For all Toeplitz words $t$ of length $\ell \geq 4$ 
coding a length-$3^\ell$ word, there exists a boolean factor of $T(t)$ of period $3^{\ell-3}$
and length $3^{\ell-2}-1$, and hence of exponent $3-3^{3-\ell}$.
\label{lemma7}
\end{lemma}

\begin{proof}
We can use {\tt Walnut} to prove this, too. Again, we let $x = 3^{|t|}$, so the hypothesis $|t| \geq 4$ is expressed as $x \geq 81$.  
\begin{verbatim}
eval critexp "?lsd_3 At,x,p ($link(x,t) & p=x/27 & x>=81) =>
Ei i+3*p<=x & $boolean(i,3*p-1,t) & $faceq(i,i+p,2*p-1,t)":
# returns TRUE
# 116 ms
\end{verbatim}
and the output from Walnut is {\tt TRUE}.

\end{proof}
This completes the proof
of Theorem~\ref{thm5}.
\end{proof}

\section{Orders of squares}

As we have seen in the previous section, each infinite Stewart word contains factors of exponent arbitrarily close to $3$, and
in particular squares.  However, the orders of the possible squares
that can occur are greatly restricted, as the following result
shows.
\begin{theorem}
Every infinite Stewart word contains squares $xx$ with
$$|x| \in \{ 3^i \suchthat i \geq 0 \} \, \cup \, 
\{ 2 \cdot 3^i \suchthat i \geq 0 \},$$ and these are
the only possible orders of squares that occur.
\end{theorem}

\begin{proof}
The proof has two steps.  First, we compute the possible orders
of squares that occur in {\it any\/} Stewart word, with
the following {\tt Walnut} code.
\begin{verbatim}
def squareorder "?lsd_3 Et,i (n>=1) & $faceq(i,i+n,n,t)":
\end{verbatim}

The resulting automaton is depicted in Figure~\ref{fig4}, and
clearly accepts only the base-$3$ representations of the form
$0^*\{ 1, 2\}$.
\begin{figure}[H]
\begin{center}
\includegraphics[width=4in]{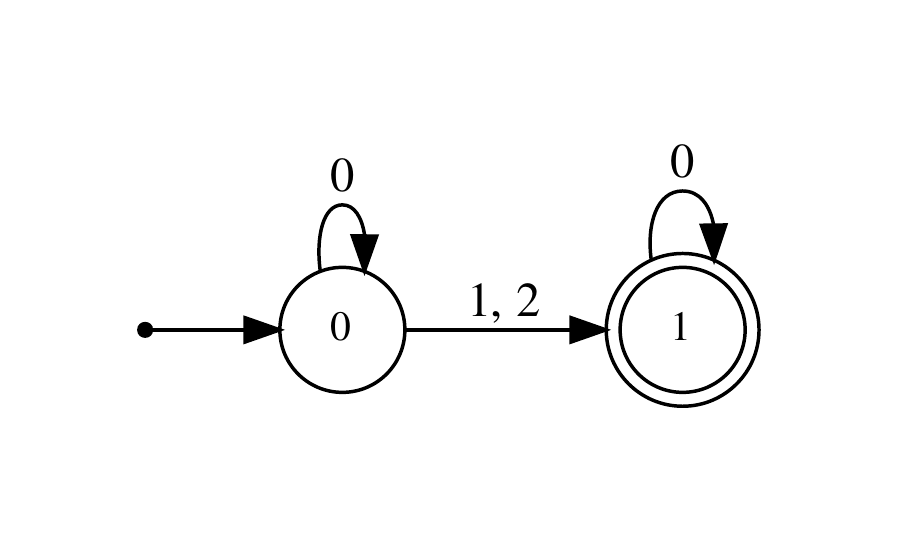}
\end{center}
\caption{Orders of squares that occur in Toeplitz words.}
\label{fig4}
\end{figure}

Next, we show that if we go far enough out in any given Stewart word, squares of the given orders actually do occur.  Using Theorem~\ref{thm3}, we see that if
a square of order $2\cdot 3^i$ occurs, it must occur in a prefix coded by
a word of length $t$, where $36 \cdot 3^i \leq 3^{|t|}$.   
\begin{verbatim}
eval all_squares_exist "?lsd_3 An,x,t ($link(x,t) & n<=x/36 & 
   Ey $power3(y) & (n=y|n=2*y)) => Ei $faceq(i,i+n,n,t)":   
\end{verbatim}
and {\tt Walnut} returns {\tt TRUE}.
\end{proof}

\section{Subword complexity}

In this section we discuss the subword complexity (aka factor complexity) of the infinite Stewart words.  The subword complexity
function $\rho(n)$ counts the number of distinct length-$n$
subwords occurring in an infinite word.   

\begin{theorem}
Every infinite word $T({\bf t})$ has exactly $2n$ distinct factors
of length $n$, for $n \geq 1$.
\end{theorem}

\begin{proof}
A factor $x$ of an infinite binary word $\bf w$ is {\it right-special\/} if both $x0$ and $x1$ occur in $\bf w$.  To
prove the theorem,
we prove the equivalent result that there are exactly $2$ distinct
right-special boolean factors of length $n$, for all $n \geq 1$.  Let
${\bf t} \in A^\omega$.  
From Theorem~\ref{thm3} we know that if a factor of length $n$ occurs in $T({\bf t})$, it occurs in a prefix $t$ of length
$3^{\lceil n \rceil}$ of $\bf t$.
We then prove that for each finite word $t$,
if we look at the lengths of factors guaranteed to occur in $T(t)$ by 
Theorem~\ref{thm3} above, then there are always exactly two such that are right-special, and there
are never three.

     First we need to define a logical formula.
\begin{itemize}
    \item $\rtspec(i,n,t)$ expresses the assertion that
    $T(t)[i..i+n-1]$ consists entirely of boolean symbols
    and is right-special for the finite word $T(t)$.
\end{itemize}

Next, we create assertions that there are at least two
right-special factors of each length $\geq 1$ and there
are never three.
\begin{verbatim}
def rtspec "?lsd_3 $boolean(i,n,t) & Ej $boolean(j,n,t) & 
   (TP[t][i+n]!=TP[t][j+n]) & $faceq(i,j,n,t)":
# vars i,n,t
# 47 states, 234 ms

eval twors "?lsd_3 At,x,n ($link7(x,t) & n>=1 & n<=x/9) => 
   Ei1,i2 $rtspec(i1,n,t) & $rtspec(i2,n,t) & ~$faceq(i1,i2,n,t)":
# returns TRUE
# 178 ms

eval threers "?lsd_3 Ei1,i2,i3,t,n (n>=1) & $rtspec(i1,n,t) & $rtspec(i2,n,t) & 
$rtspec(i3,n,t) & ~$faceq(i1,i2,n,t) & ~$faceq(i2,i3,n,t) & ~$faceq(i1,i3,n,t)":
\end{verbatim}
Here {\tt twors} returns {\tt TRUE} and {\tt threers} returns
{\tt FALSE}; 
this completes the proof.
\end{proof}

Thus, the Stewart words provide yet more examples of the
class of sequences studied by Rote \cite{Rote:1994}.

\section{Common factors}

Recall that
$A = \{ {\tt a,b,c,d,e,f} \}$.
Let ${\bf t,u} \in A^\omega$.   We would like to characterize those pairs $({\bf t}, {\bf u})$ such that the Stewart words $T({\bf t})$ and $T({\bf u})$ contain arbitrarily large factors in
common.   

To do so we need to introduce the Hamming distance between
Stewart patterns.     We define $H(t,u)$ for
$t,u \in A$ as the number of positions
on which $T(t)$ and $T(u)$ differ.  For example,
$H({\tt a},{\tt b}) = 2$, since the patterns {\tt 01?} and
{\tt 10?} differ in the first two positions, but agree on the
third.

Next, define 
\begin{align*}
    X &:= \{ (t,u) \in A\times A \suchthat H(t,u) \in \{ 0,3 \} \};\\
    Y &:= \{ (t,u) \in A\times A \suchthat H(t,u)=2 \}.
\end{align*}
Then we have the following result.
\begin{theorem}
The infinite Stewart words $T({\bf t})$ and $T({\bf u})$ share
arbitrarily large factors in common if and only if
${\bf t} \times {\bf u} \in X^\omega$.    Furthermore, if $j$ is the
smallest index on which $({\bf t}\times {\bf u})[j] \not\in X$,
then $T({\bf t})$ and $T({\bf u})$ have no factors in common of
length $> 3^{j+2}$.
\end{theorem}

\begin{proof}
For the first claim, we use the following {\tt Walnut} code.
\begin{verbatim}
eval commonfac "?lsd_3 Ex,i1,i2 $link(x,t) & $link(x,u) & $cmp(i1,i2,x/9,t,u)":
# 4 states, 350259ms
\end{verbatim}
{\tt Walnut} then computes a
$4$-state automaton accepting 
those pairs of finite sequences having common factors.  In abbreviated form, this is depicted in Figure~\ref{cf}.
\begin{figure}[H]
\begin{center}
\includegraphics[width=5in]{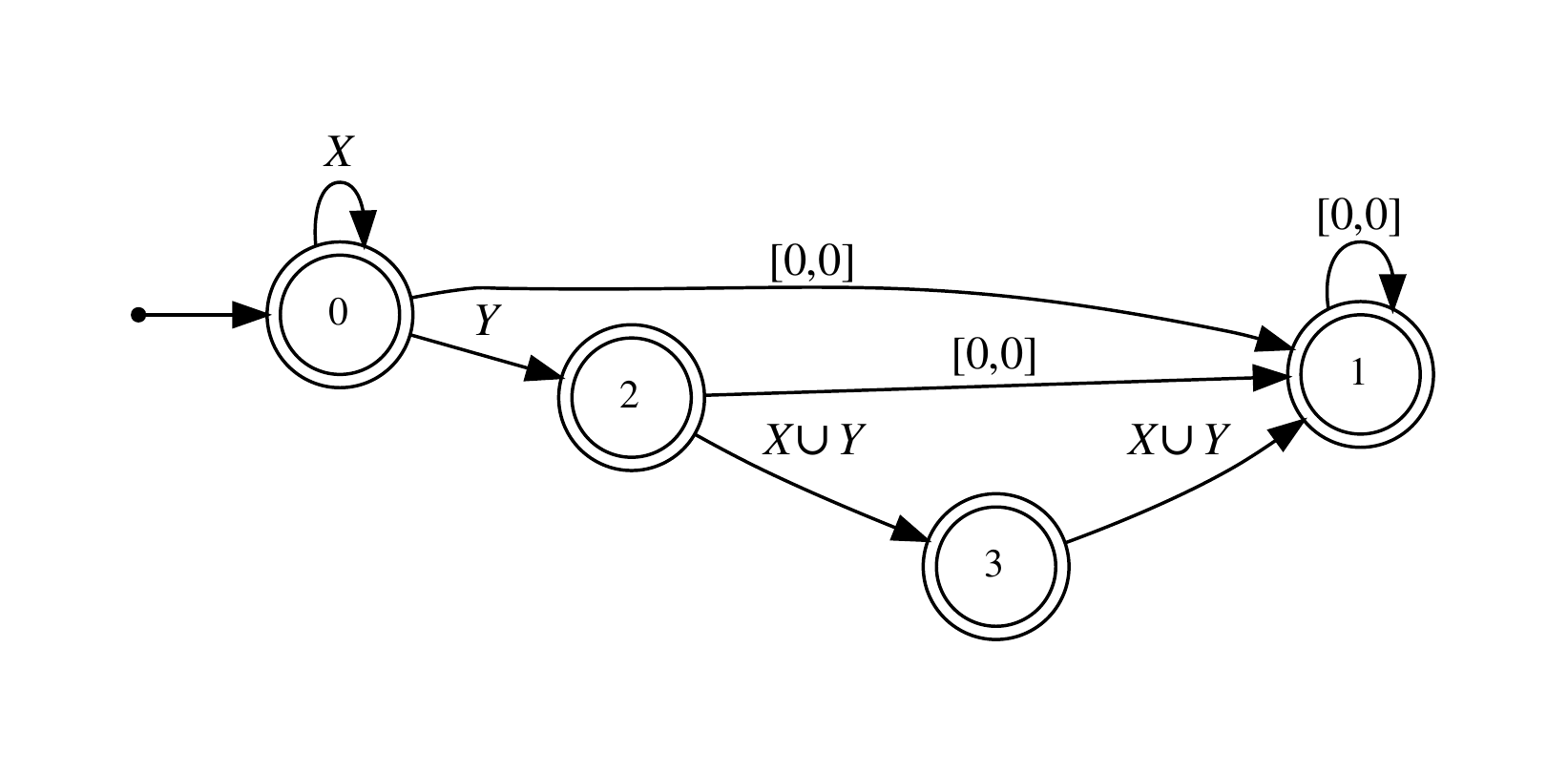}
\end{center}
\caption{Finite sequences having common factors.}
\label{cf}
\end{figure}
By inspection, we see that the only {\it infinite\/} sequences accepted by
this automaton are those in
$X^\omega$.

For the second claim, we need to build an automaton that identifies the first position where two finite
words $t, u$ of the same length differ.  This can be done in {\tt Walnut} as follows:
\begin{verbatim}
reg differ lsd_7 lsd_7 lsd_3
"([1,1,0]|[1,4,0]|[1,5,0]|[2,2,0]|[2,3,0]|[2,6,0]|
[3,2,0]|[3,3,0]|[3,6,0]|[4,1,0]|[4,4,0]|[4,5,0]|
[5,1,0]|[5,4,0]|[5,5,0]|[6,2,0]|[6,3,0]|[6,6,0])*
([1,2,1]|[1,3,1]|[1,6,1]|[2,1,1]|[2,4,1]|[2,5,1]|
[3,1,1]|[3,4,1]|[3,5,1]|[4,2,1]|[4,3,1]|[4,6,1]|
[5,2,1]|[5,3,1]|[5,6,1]|[6,1,1]|[6,4,1]|[6,5,1])
([1,1,0]|[1,2,0]|[1,3,0]|[1,4,0]|[1,5,0]|[1,6,0]|
[2,1,0]|[2,2,0]|[2,3,0]|[2,4,0]|[2,5,0]|[2,6,0]|
[3,1,0]|[3,2,0]|[3,3,0]|[3,4,0]|[3,5,0]|[3,6,0]|
[4,1,0]|[4,2,0]|[4,3,0]|[4,4,0]|[4,5,0]|[4,6,0]|
[5,1,0]|[5,2,0]|[5,3,0]|[5,4,0]|[5,5,0]|[5,6,0]|
[6,1,0]|[6,2,0]|[6,3,0]|[6,4,0]|[6,5,0]|[6,6,0])*
[0,0,0]*":
\end{verbatim}
Here $\differ(t,u,x)$ is true if $x = 3^j$ and
$(t[j],u[j]) \not\in X$.    Then we use
the {\tt Walnut} code 
\begin{verbatim}
eval compare "?lsd_3 At,u,x ($link(243*x,t) & $link(243*x,u) & $differ(t,u,x))
   => Ai Aj ~$cmp(i,j,9*x+1,t,u)":
# TRUE, 14401 ms
\end{verbatim}
\noindent which returns {\tt TRUE}.
\end{proof}

\section{Automatic Stewart words}

A paperfolding word is automatic if and only if its associated sequence of patterns is ultimately periodic (see \cite{Allouche&Shallit:2003}). An analogous result holds for Stewart words:

\begin{theorem}
A Stewart word $T({\bf t})$ is $3$-automatic if and only if $\bf t$ is ultimately periodic.
\end{theorem}

\begin{proof}
     Suppose $\bf t$ is ultimately periodic.   Then the language $L$ of prefixes of $\bf t$ is regular.   We can then use the usual direct product construction for automata to create an automaton $M'$ from $M$
     by forcing the first components of inputs to lie in $L$. 
     Then $M'$ demonstrates that $T({\bf t})$ is
     $3$-automatic.
     
On the other hand, suppose  $T({\bf t})$ is $3$-automatic and computed by a DFAO 
     $$N = (Q, \Sigma_3, \Delta, \delta, q_0, \tau).$$ 
     It is easy to check that $t_0 := {\bf t}[0]$
     is uniquely specified by
     the values of $T({\bf t})[0..8]$; that is,
     by $\tau(\delta(q_0, x))$ for $|x|= 2$.
     Next, locate the position $p_0$
     of the {\tt ?} symbol in $t_0$, and determine
     $q_1 = \delta(q_0, p_0)$.   By examining the outputs of $N$ on words of length $2$, 
     starting from the state
     $q_1$, we can then uniquely determine $t_1$.   Continuing in this fashion,
     we generate ${\bf t}$ element by
     element.   Since $N$ has finitely many states,
     eventually we return to a state previously
     visited, at which point $\bf t$ becomes
     periodic.
\end{proof}

\begin{example}
As an example, consider the $3$-automatic sequence
computed by the DFAO (lsd-first) in Figure~\ref{fig7}.
\begin{figure}[H]
\begin{center}
\includegraphics[width=4in]{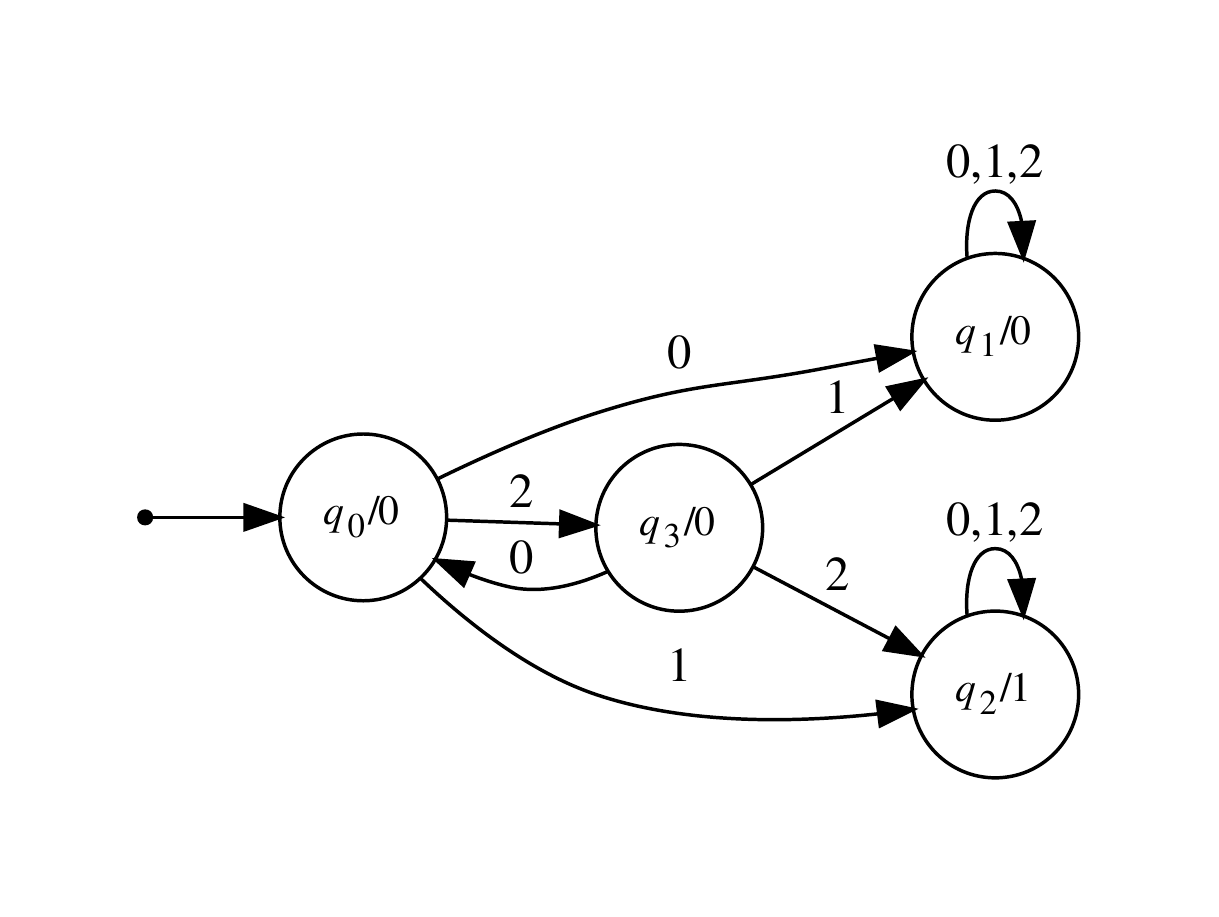}
\end{center}
\caption{A $3$-automatic sequence.}
\label{fig7}
\end{figure}
Our labeling procedure
labels $q_0$ with $\tt a$ and then $q_3$ with
$\tt d$, and then returns to state $q_0$.  So
${\bf t} = {\tt (ad)}^\omega$.
\end{example}

\section{Arithmetic progressions}

We are concerned with occurrences
of the words {\tt 01010} and
{\tt 10101} in arithmetic progressions in an infinite word $\bf x$; this means there
exist indices $i, m$ with
$m \geq 1$ such that
${\bf x}[i]={\bf x}[i+2m]={\bf x}[i+4m] = a$ and
${\bf x}[i+m]={\bf x}[i+3m] = \overline{a}$, for some
$a \in \{0,1\}$.

\begin{theorem}
No infinite Stewart word contains
either {\tt 01010} or
{\tt 10101} in arithmetic progression.
\end{theorem}

\begin{proof}
We can prove this 
with the following {\tt Walnut} statement asserting the existence of such an arithmetic progression.
\begin{verbatim}
eval arithprog "?lsd_3 Et,i,m (m>=1) & TP[t][i]=TP[t][i+2*m] &
TP[t][i]=TP[t][i+4*m] & TP[t][i+m]=TP[t][i+3*m] & TP[t][i]!=TP[t][i+m]":
# 607 ms
\end{verbatim}
and {\tt Walnut} returns
{\tt FALSE}.
\end{proof}

For more about arithmetic progressions in Toeplitz words, see
\cite{Hendriks&Dannenberg&Endrullis&Dow&Klop:2012}.

\section{Going further}

Everything we have done in this paper applies (with minor changes)
to any finite set of Toeplitz patterns of the same length, each
containing exactly one occurrence of the symbol {\tt ?}.  

\section*{Acknowledgments}

We thank Hamoon Mousavi, Aseem Raj Baranwal and Laindon C. Burnett for their work on {\tt Walnut} that made this paper possible.

\newpage
\appendix
\section*{Appendix}

This is the file {\tt TP.txt} defining the Stewart automaton.
The interested reader is recommended to put it in the directory
{\tt Word Automata Library} of
{\tt Walnut} before typing any
commands.

{\scriptsize
\begin{verbatim}
lsd_7 lsd_3
0 2
0 0 -> 3
1 0 -> 1
1 1 -> 2
1 2 -> 0
2 0 -> 2
2 1 -> 1
2 2 -> 0
3 0 -> 1
3 1 -> 0
3 2 -> 2
4 0 -> 2
4 1 -> 0
4 2 -> 1
5 0 -> 0
5 1 -> 1
5 2 -> 2
6 0 -> 0
6 1 -> 2
6 2 -> 1

1 0
0 0 -> 4
0 1 -> 1
0 2 -> 1
1 0 -> 1
1 1 -> 1
1 2 -> 1
2 0 -> 1
2 1 -> 1
2 2 -> 1
3 0 -> 1
3 1 -> 1
3 2 -> 1
4 0 -> 1
4 1 -> 1
4 2 -> 1
5 0 -> 1
5 1 -> 1
5 2 -> 1
6 0 -> 1
6 1 -> 1
6 2 -> 1

2 1 
0 0 -> 5
1 0 -> 2
1 1 -> 2
1 2 -> 2
2 0 -> 2
2 1 -> 2
2 2 -> 2
3 0 -> 2
3 1 -> 2
3 2 -> 2
4 0 -> 2
4 1 -> 2
4 2 -> 2
5 0 -> 2
5 1 -> 2
5 2 -> 2
6 0 -> 2
6 1 -> 2
6 2 -> 2

3 2
0 0 -> 3

4 0
0 0 -> 4

5 1
0 0 -> 5
\end{verbatim}
}

\begin{thebibliography}{10}

\bibitem{Allouche&Shallit:2003}
J.-P. Allouche and J.~Shallit.
\newblock {\em Automatic Sequences: Theory, Applications, Generalizations}.
\newblock Cambridge University Press, 2003.

\bibitem{Berstel&Karhumaki:2003}
J.~Berstel and J.~Karhum{\"a}ki.
\newblock Combinatorics on words---a tutorial.
\newblock {\em Bull. European Assoc. Theor. Comput. Sci.}, No.\ 79, (2003),
  178--228.

\bibitem{Boccuto&Carpi:2020}
A.~Boccuto and A.~Carpi.
\newblock Repetitions in {Toeplitz} words and the {Thue} threshold.
\newblock In M.~Anselmo et~al., editors, {\em CiE 2020}, Vol. 12098 of {\em
  Lecture Notes in Computer Science}, pp.  264--276. Springer-Verlag, 2020.

\bibitem{Bruyere&Hansel&Michaux&Villemaire:1994}
V.~{Bruy\`ere}, G.~Hansel, C.~Michaux, and R.~Villemaire.
\newblock Logic and $p$-recognizable sets of integers.
\newblock {\em Bull. Belgian Math. Soc.} {\bf 1} (1994), 191--238.
\newblock Corrigendum, {\it Bull.\ Belg.\ Math.\ Soc.} {\bf 1} (1994), 577.

\bibitem{Cassaigne&Karhumaki:1997}
J.~Cassaigne and J.~{Karhum\"aki}.
\newblock Toeplitz words, generalized periodicity and periodically iterated
  morphisms.
\newblock {\em European J. Combin.} {\bf 18} (1997), 497--510.

\bibitem{Dekking&MendesFrance&vanderPoorten:1982}
F.~M. Dekking, M.~{Mend\`es}~France, and A.~J. van~der Poorten.
\newblock Folds!
\newblock {\em Math. Intelligencer} {\bf 4} (1982), 130--138, 173--181,
  190--195.
\newblock Erratum, {\bf 5} (1983), 5.

\bibitem{Fici:2021}
G.~Fici.
\newblock Some remarks on automatic sequences, {Toeplitz} words and perfect
  shuffling.
\newblock Talk for the One World Combinatorics on Words Seminar. Available at
  \url{http://www.i2m.univ-amu.fr/wiki/Combinatorics-on-Words-seminar/_media/seminar2021:20211206fici.pdf},
  December 6 2021.

\bibitem{Goc&Mousavi&Schaeffer&Shallit:2015}
D.~Go\v{c}, H.~Mousavi, L.~Schaeffer, and J.~Shallit.
\newblock A new approach to the paperfolding sequences.
\newblock In A.~Beckmann et~al., editor, {\em Computability in Europe, Cie
  2015}, Vol. 9136 of {\em Lecture Notes in Computer Science}, pp.  34--43.
  Springer-Verlag, 2015.

\bibitem{Hendriks&Dannenberg&Endrullis&Dow&Klop:2012}
D.~Hendriks, F.~G.~W. Dannenberg, J.~Endrullis, M.~Dow, and J.~W. Klop.
\newblock Arithmetic self-similarity of infinite sequences.
\newblock Arxiv preprint 1201.3786v3. Available at
  \url{https://arxiv.org/abs/1201.3786}, 2012.

\bibitem{Jacobs&Keane:1969}
K.~Jacobs and M.~Keane.
\newblock $0$--$1$-sequences of {Toeplitz} type.
\newblock {\em Z. Wahrscheinlichkeitstheorie und Verw. Gebiete} {\bf 13}
  (1969), 123--131.

\bibitem{Mousavi:2016}
H.~Mousavi.
\newblock Automatic theorem proving in {{\tt Walnut}}.
\newblock Arxiv preprint, arXiv:1603.06017 [cs.FL], available at
  \url{http://arxiv.org/abs/1603.06017}, 2016.

\bibitem{Noche:2008}
J.~Reyes Noche.
\newblock On {Stewart's} choral sequence.
\newblock {\em Gib{\'on}} {\bf 8}(1) (2008), 1--5.

\bibitem{Noche:2008b}
J.~Reyes Noche.
\newblock Generalized choral sequences.
\newblock {\em Matimy{\'a}s Matematika} {\bf 31} (2008), 25--28.

\bibitem{Noche:2011}
J.~Reyes Noche.
\newblock On generalized choral sequences.
\newblock {\em Gib{\'on}} {\bf 9} (2011), 51--69.

\bibitem{Rote:1994}
G.~Rote.
\newblock Sequences with subword complexity $2n$.
\newblock {\em J. Number Theory} {\bf 46} (1994), 196--213.

\bibitem{Sloane:2021}
N.~J.~A. Sloane et~al.
\newblock The on-line encyclopedia of integer sequences.
\newblock Electronic resource, available at \url{https://oeis.org}, 2021.

\bibitem{Stewart:2006}
I.~Stewart.
\newblock {\em How to Cut a Cake: And Other Mathematical Conundrums}.
\newblock Cambridge University Press, 2006.

\end{thebibliography}
\end{document}